\documentclass[12pt]{article}
\usepackage{fullpage}
\usepackage{graphicx}
\usepackage{amsmath}
\usepackage{amssymb}
\usepackage{amsthm}

\renewcommand{\v}{\mathbf{v}}
\newcommand{\cross}{+}
\newcommand{\im}{\mathrm{im}}

\newcommand{\onevec}{\mathbf{1}}
\newcommand{\R}{\mathbb{R}}

\newtheorem{theorem}{Theorem}[section]

\newtheorem{lemma}[theorem]{Lemma}

\newtheorem{proposition}[theorem]{Proposition}
\newtheorem{definition}[theorem]{Definition}
\newtheorem{claim}[theorem]{Claim}
\newtheorem{conjecture}[theorem]{Conjecture}

\def\softO#1{\widetilde{{O}} \left( #1 \right)}
\def\bigO#1{{{O}}\left( #1 \right) }

\renewcommand{\v}{\mathbf{v}}

\newcommand{\lmin}{\lambda_\mathrm{min}}
\newcommand{\lmax}{\lambda_\mathrm{max}}

\newcommand{\deltaL}{\delta_{L}}
\newcommand{\deltaU}{\delta_{U}}
\newcommand{\epsL}{\epsilon_{L}}
\newcommand{\epsU}{\epsilon_{U}}
\newcommand{\tr}{\mathrm{Tr}}
\newcommand{\id}{\mathbf{id}}

\newcommand\ppi{\mathbf{v}}

\def\ceil#1{\left\lceil #1 \right\rceil}

\def\norm#1{\left\| #1 \right\|}

\newdimen\pIR
\newcommand\StevesR{{\rm I\kern\pIR R}}

\def\bvec#1{{\mbox{\boldmath $#1$}}}

\def\defeq{\stackrel{\mathrm{def}}{=}}
\def\setof#1{\left\{#1  \right\}}

\def\sizeof#1{\left|#1  \right|}

\begin{document}

\title{Twice-Ramanujan Sparsifiers
\thanks{
This material is based upon work supported by the National Science Foundation under Grant CCF-0634957.
Any opinions, findings, and conclusions or recommendations expressed in this material are those of the author(s) and do not necessarily reflect the views of the National Science Foundation.
}\\} 
 
\author{
Joshua Batson
\thanks{
Department of Mathematics, MIT. Work on this paper performed while at
Yale College.}\\
\and
Daniel A. Spielman
\thanks{
Program in Applied Mathematics and 
Department of Computer Science,
Yale University.}\\
\and
Nikhil Srivastava
\thanks{
Department of Computer Science,
Yale University.}
}

\maketitle
\begin{abstract}
We prove that every graph has a spectral sparsifier with a number
  of edges linear in its number of vertices.
As linear-sized spectral sparsifiers of complete graphs are expanders,
  our sparsifiers of arbitrary graphs 
  can be viewed as generalizations of
  expander graphs.

In particular, we prove that for every $d > 1$ and every
  undirected, weighted graph $G = (V,E,w)$ on $n$ vertices, there exists
  a weighted graph $H=(V,F,\tilde{w})$ with at most $\ceil{d(n-1)} $ edges
  such that for every $x \in \R^{V}$,
\[
 x^T L_G x \leq 
        x^{T} L_{H} x
    \leq
\left(\frac{d+1+2\sqrt{d}}{d+1-2\sqrt{d}}\right)
      \cdot
      x^{T} L_{G} x
\]
where $L_{G}$ and $L_{H}$ are the Laplacian matrices of $G$
  and $H$, respectively. 
Thus, $H$
  approximates $G$ spectrally at least as well as a Ramanujan expander
  with $dn/2$ edges approximates the complete graph.

We give an elementary deterministic polynomial time algorithm for
  constructing $H$.
\end{abstract}

\section{Introduction}\label{sec:intro}
A sparsifier of a graph $G = (V,E,w)$ is a sparse graph $H$ that is similar to
  $G$ in some useful way.
Many notions of similarity have been considered.
For example, Chew's~\cite{PaulChew} spanners have the property that the distance
  between every pair of vertices in $H$ is approximately the same as in $G$.
Benczur and Karger's~\cite{BenczurKarger} cut-sparsifiers have the property
  that the weight of the boundary of every set of vertices is approximately
  the same in $G$ as in $H$.
We consider the spectral notion of similarity introduced
  by Spielman and Teng~\cite{SpielmanTengPrecon,SpielmanTengSparsifier}:
  we say that $H$ is a $\kappa$-approximation of $G$ if for all
  $x \in \R^{V}$,
\begin{equation}\label{eqn:approximation}
 x^{T} L_G x \leq 
        x^{T} L_{H} x
  \leq
  \kappa\cdot
        x^{T} L_{G} x,
\end{equation}
where $L_{G}$ and $L_{H}$ are the Laplacian matrices of $G$ and $H$.
We recall that
\[
x^{T} L_{G} x = 
\sum_{(u,v) \in E} w_{u,v} (x_{u} - x_{v})^{2},
\]
where $w_{u,v}$ is the weight of edge $(u,v)$ in $G$.
By considering vectors $x$ that are the characteristic vectors of sets,
  one can see that condition~\eqref{eqn:approximation}
  is strictly stronger than the cut
  condition of Benczur and Karger.

In the case where $G$ is the complete graph, excellent spectral
  sparsifiers are supplied
  by \textit{Ramanujan Graphs} \cite{LPS,Margulis}.
These are $d$-regular graphs $H$ 
  all of whose non-zero Laplacian eigenvalues lie between
  $d - 2 \sqrt{d-1}$ and $d + 2 \sqrt{d-1}$.
Thus, if we take a Ramanujan graph on $n$ vertices and
  multiply the weight of every edge by
  $n / (d - 2 \sqrt{d-1})$, we obtain a graph
  that
   $\kappa$-approximates
  the complete graph, for 
\[
  \kappa  = \frac{d+2\sqrt{d-1}}{d-2\sqrt{d-1}}.
\]
  
In this paper, we prove that
  every graph can be approximated at least this 
  well\footnote{
Strictly speaking, our approximation constant is only better than the Ramanujan
  bound $\kappa=\frac{d+2\sqrt{d-1}}{d-2\sqrt{d-1}}$ in the regime $d\ge
  \frac{1+\sqrt{5}}{2}$.
  This includes the actual Ramanujan graphs, for which $d$ is an integer greater
  than $2$.}
     by a graph with only twice as many edges 
  as the Ramanujan graph
  (as a $d$-regular graph has $dn/2$ edges).
\begin{theorem}\label{thm:mainthm}
For every $d>1$, 
  every undirected weighted graph $G=(V,E,w)$ on $n$ vertices
  contains a weighted subgraph $H=(V,F,\tilde{w})$ with $\ceil{d(n-1)}$
  edges (i.e., average degree at most $2d$) that satisfies:
\[
 x^T L_G x \leq 
        x^{T} L_{H} x
    \leq
\left(\frac{d+1+2\sqrt{d}}{d+1-2\sqrt{d}}\right)
      \cdot
      x^{T} L_{G} x
      \qquad\forall x\in\R^V.
\]
\end{theorem}
Our proof provides
  a deterministic greedy algorithm for computing the graph $H$ in 
  time   $O (d n^{3} m)$.

We remark that while the edges of $H$ are a subset of the edges of $G$,
  the weights of edges in $H$ and $G$ will typically be different.
In fact, there exist unweighted graphs $G$ for which every good spectral 
  sparsifier $H$ must contain edges of 
  widely varying weights~\cite{SpielmanTengSparsifier}.

\subsection{Expanders: Sparsifiers of the Complete Graph}
In the case that $G$ is a complete graph, 
  our construction produces expanders.
However, these expanders are slightly unusual in that their edges have weights,
  they may be irregular, and the weighted degrees of vertices can vary slightly.
This may lead one to ask whether they should really be considered expanders.
In Section~\ref{sec:expanders} we argue that they should be.

As the graphs we produce are irregular and weighted, it is also not immediately
  clear that we should be comparing $\kappa$ with the Ramanujan bound of
\begin{equation}\label{eqn:ramanujanBound}
\frac{
        d + 2 \sqrt{d-1}
}{
        d - 2 \sqrt{d-1}
}
=
1+\frac{4}{\sqrt{d}}+O(1/d).
\end{equation}
It is known%
\footnote{While lower bounds on the spectral gap of $d$-regular graphs
  focus on showing that the second-smallest eigenvalue is asymptotically
  at most
  $d - 2 \sqrt{d-1}$, the same proofs by test functions
  can be used to show that the
  largest eigenvalue is at asymptotically least $d + 2 \sqrt{d-1}$.}
 that no $d$-regular graph of uniform weight
  can $\kappa$-approximate a complete graph for $\kappa$ 
  asymptotically better than 
  \eqref{eqn:ramanujanBound}~\cite{Nilli}.
While we believe that no graph of {\em average} degree $d$ can be a
  $\kappa$-approximation of a complete graph for $\kappa$ asymptotically
  better than \eqref{eqn:ramanujanBound}, we are unable to show this at the
  moment and prove
  instead the weaker claim that no such graph can achieve $\kappa$ 
  less than
\[
1 + 
\frac{2}{\sqrt{d}}
- \bigO{\frac{\sqrt{d}}{n}}.
\]

\subsection{Prior Work}\label{sec:prior}
Spielman and Teng~\cite{SpielmanTengPrecon,SpielmanTengSparsifier} 
  introduced the notion of
  sparsification that we consider, and proved that $(1+\epsilon )$-approxim\-ations
  with $\softO{n / \epsilon^{2}}$ edges could be constructed in $\softO{m}$ time.
They used these sparsifiers to obtain a nearly-linear time algorithm
  for solving diagonally dominant systems of linear 
  equations~\cite{SpielmanTengPrecon,SpielmanTengLinsolve}.

Spielman and Teng were inspired by the notion of sparsification introduced by Benczur
  and Karger~\cite{BenczurKarger} for cut problems, which only
  required 
  inequality \eqref{eqn:approximation} to hold for all $x \in \setof{0,1}^{V}$.
Benczur and Karger showed how to construct graphs $H$ meeting this guarantee
  with  $\bigO{n \log n / \epsilon ^{2}}$ edges in $\bigO{m \log^{3} n}$ time;
 their cut sparsifiers have been used to obtain faster algorithms
  for cut problems~\cite{BenczurKarger,krv}.

Spielman and Srivastava~\cite{SpielmanSrivastava} proved the
  existence of spectral sparsifiers with $\bigO {n \log n / \epsilon ^{2}}$
  edges, and showed how to construct them in $\softO{m}$ time.
They conjectured that it should be possible to find such sparsifiers with only
  $\bigO{n / \epsilon ^{2}}$ edges.
We affirmatively resolve this conjecture.

Recently, partial progress was made towards this conjecture
  by Goyal, Rademacher and Vempala~\cite{GoyalRademacherVempala},
  who showed how to find graphs $H$ with only $2n$ edges that 
  $\bigO {\log n}$-approximate bounded degree graphs
  $G$ under the cut notion of Benczur and Karger.

We remark that all of these constructions were randomized.
Ours is the first deterministic algorithm to achieve the guarantees of any
  of these papers.

\section{Preliminaries}
\subsection{The Incidence Matrix and the Laplacian} \label{sec:incidence}
Let $G=(V,E,w)$ be a connected weighted undirected graph with $n$ vertices and $m$
edges and edge weights $w_e > 0$. If we orient the edges of $G$
arbitrarily, we can write its Laplacian as $L=B^{T} WB$, where $B_{m\times n}$ is the
{\em signed edge-vertex incidence matrix}, given by
\[ B(e,v)= \left\{\begin{array}{ll} 1 &\textrm{if $v$ is $e$'s head}\\ -1 &\textrm{if
$v$ is $e$'s tail}\\ 0 & \textrm{otherwise}\end{array}\right.\]
and $W_{m\times m}$ is the diagonal matrix with $W(e,e)=w_e$. 
It is immediate that $L$ is positive semidefinite since:
\begin{align*}
 x^{T} Lx &=
  x^{T} B^{T} WBx =
 \|W^{1/2}Bx\|_2^2
 \\&=\sum_{(u,v) \in E} w_{u,v} (x_{u} - x_{v})^{2}
 \ge 0,
\quad\textrm{ for every $x\in\R^n$.}\end{align*}
\noindent and that $G$ is connected if and only if
 $\ker(L)=\ker(W^{1/2}B)=\textrm{span}(\onevec)$.

\subsection{The Pseudoinverse}\label{sec:pseudo}
Since $L$ is symmetric we can diagonalize it and write
\[ L=\sum_{i=1}^{n-1}\lambda_i u_iu_i^{T} \]
where $\lambda_1,\ldots,\lambda_{n-1}$ are the nonzero eigenvalues of $L$ and
$u_1,\ldots,u_{n-1}$ are a corresponding set of orthonormal eigenvectors. The {\em
Moore-Penrose Pseudoinverse} of $L$ is then defined as
\[ L^\cross=\sum_{i=1}^{n-1}\frac{1}{\lambda_i} u_iu_i^{T} .\]
Notice that $\ker(L)=\ker(L^\cross)$ and that
\[ LL^\cross=L^\cross L=\sum_{i=1}^{n-1}u_iu_i^{T} ,\]
which is simply the projection onto the span of the
nonzero eigenvectors of $L$ (which are also the eigenvectors of $L^\cross$). Thus, $LL^\cross=L^\cross L$ is the identity on
$\im(L)=\ker(L)^\perp$.

\subsection{Formulas for Rank-one Updates}
We use the following well-known theorem from linear algebra, which describes
the behavior of the inverse of a matrix under rank-one updates 
  (see~\cite[Section 2.1.3]{GolubVanLoan}).
\begin{lemma}[Sherman-Morrison Formula] If $A$ is a nonsingular $n\times n$ matrix and $\ppi$ is a vector, then 
\[(A+\ppi\ppi^{T} )^{-1}= A^{-1}-\frac{A^{-1}\ppi\ppi^{T} A^{-1}}{1+\ppi^{T} A^{-1}\ppi}.\]
\end{lemma}
There is a related formula describing the change in the {\em determinant} of
  a matrix under the same update:
\begin{lemma}[Matrix Determinant Lemma] \label{lem:matrixdet} If $A$ is nonsingular and $\ppi$ is a vector, then
\[ \det(A+\ppi\ppi^T) = \det(A)(1+\ppi^T A^{-1} \ppi).\]\end{lemma}

\section{The Main Result}
At the heart of this work is the following purely linear algebraic theorem.
We use the notation $A\preceq B$ to mean that $B-A$ is positive
  semidefinite, and $\id_{S}$ to denote the identity operator on a vector space $S$.
\begin{theorem}\label{thm:linalg}
Suppose $d>1$ and $\v_1,\v_2,\ldots,\v_m$ are vectors in $\R^n$ with
\[
  \sum_{i\le m}\v_i\v_i^T=\id_{\R^n}.
\]
Then there exist scalars $s_i\ge 0$ with $|\{i:s_i\neq 0\}|\le dn$ so that
\[
  \id_{\R^n} 
  \preceq \sum_{i\le m} s_i\v_i\v_i^T
  \preceq \left(\frac{d+1+2\sqrt{d}}{d+1-2\sqrt{d}}\right) \id_{\R^n}.
\]
\end{theorem}
\noindent The sparsification result for graphs follows quickly from this theorem as shown below.
\begin{proof}[Proof of Theorem \ref{thm:mainthm}]
Assume without loss of generality that $G$ is connected. Write $L_G=B^TWB$ as in Section \ref{sec:incidence} and fix $d>1$. 
Restrict attention to $\im(L_G)\cong \R^{n-1}$ and apply Theorem \ref{thm:linalg} to the columns $\{\v_i\}_{i\le m}$ of 
\[
  V_{n\times m} = (L_G^+)^{\frac{1}{2}}B^TW^{\frac{1}{2}},
\] 
which are indexed by the edges of $G$ and satisfy
\begin{align*}
  \sum_{i\le m} \v_i\v_i^T = VV^T &= (L_G^+)^{\frac{1}{2}}B^TWB(L_G^+)^{\frac{1}{2}}
  \\&= (L_G^+)^{\frac{1}{2}}L_G(L_G^+)^{\frac{1}{2}} = \id_{\im(L_G)}.
\end{align*}
Write the scalars $s_i\ge 0$ guaranteed by the theorem in the $m\times m$ diagonal
matrix $S(i,i)=s_i$ and set $L_H=B^TW^{\frac{1}{2}}SW^{\frac{1}{2}}B$. Then $L_H$ is the Laplacian of
the subgraph $H$ of $G$ with edge weights $\{\tilde{w}_i=w_i s_i\}_{i\in E}$, and $H$ has at most $d(n-1)$
edges since at most that many of the $s_i$ are nonzero. Also,

\begin{align*} 
  &\id_{\im(L_G)}\preceq 
  \sum_{i\le m}s_i\v_i\v_i^T = VSV^T 
  \preceq \kappa \cdot \id_{\im(L_G)}
  \qquad\textrm{for $\kappa=\frac{d+1+2\sqrt{d}}{d+1-2\sqrt{d}}$.}
\end{align*}
By the Courant-Fischer Theorem, this is equivalent to:
\begin{align*}
  1&\le \frac{
      y^TVSV^Ty
    }{
      y^Ty
    } 
    \le 
    \kappa 
 \qquad \forall y\in\im((L_G)^\frac{1}{2})=\im(L_G)&
\\\iff &1\le 
    \frac{
      y^T(L_G^+)^{\frac{1}{2}}L_H(L_G^+)^{\frac{1}{2}}y
    }{
      y^Ty
    }
  \le 
  \kappa
\qquad\forall y\in\im((L_G)^{\frac{1}{2}})&
\\\iff &1\le 
    \frac{
      x^TL_G^{\frac{1}{2}}(L_G^{+})^{\frac{1}{2}}L_H(L_G^+)^{\frac{1}{2}}L_G^{\frac{1}{2}}x
    }{
      x^TL_G^{\frac{1}{2}}L_G^{\frac{1}{2}}x
    }
    \le \kappa\quad\qquad\forall x\perp \mathbf{1}
\\\iff &1\le 
    \frac{
      x^TL_Hx
    }{
      x^TL_Gx
    }\le \kappa\quad\qquad\qquad\qquad\qquad\qquad\forall x\perp \mathbf{1},
\end{align*}
as desired.
\end{proof}

It is worth mentioning that the above reduction is essentially the same as
  the one in \cite{SpielmanSrivastava}. 
In that paper, the authors consider the
  symmetric projection matrix $\Pi = BL^{+}_GB^T$ whose columns $\{\Pi_e\}_{e\in
  E}$
  correspond to the edges of $G$.
They show, by a concentration lemma of Rudelson~\cite{rudl},
  that randomly sampling $O(n\log n)$ of the columns with probabilities proportional to
  $\|\Pi_e\|^2=w_e\mathsf{R_{eff}}(e)$ (where $\mathsf{R_{eff}}$ is the effective
  resistance)
  gives a matrix $\tilde{\Pi}$ that approximates $\Pi$ in the
  spectral norm and corresponds to a graph sparsifier, with high probability.
In this paper, we do essentially the same thing with two modifications:
  we eliminate $\Pi$ in order to simplify notation, since we are no longer
  following the intuition of sampling by effective resistances; and, instead of
  Rudelson's 
  sampling lemma, we use Theorem \ref{thm:linalg} to {\em deterministically}
  select $O(n)$ edges (equivalently, columns of $\Pi$).

The rest of this section is devoted to proving Theorem \ref{thm:linalg}.
The proof is constructive and yields a deterministic polynomial time algorithm for finding the
  scalars $s_i$, which can then be used to sparsify graphs, as advertised.

Given vectors $\{\v_i\}$, our goal is to choose a small set of coefficients $s_i$ so that 
  $A=\sum_i s_i\ppi_{i}\ppi_{i}^{T} $ is well-conditioned. 
We will build the matrix $A$ in steps, starting with $A=0$ and adding one vector
  $s_i\v_i\v_i^T$ at a time. 
Before beginning the proof, it will be instructive to study 
  how the eigenvalues and characteristic polynomial of a matrix evolve upon the addition of a vector.
This discussion should provide some intuition for the structure of the proof, and demystify 
  the origin of the `Twice-Ramanujan' number $\frac{d+1+2\sqrt{d}}{d+1-2\sqrt{d}}$ which 
  appears in our final result.
\subsection{Intuition for the Proof}
It is well known that the eigenvalues of $A+\ppi\ppi^T$ interlace those of $A$.
In fact, the new eigenvalues can be determined exactly by looking at the 
  characteristic polynomial of $A+\ppi\ppi^T$, which is computed using 
  Lemma \ref{lem:matrixdet} as follows:
\[ p_{A+\ppi\ppi^T}(x)=\det(xI-A-\ppi\ppi^T)=p_A(x)\left(1-\sum_j\frac{\langle
\ppi, u_j\rangle^2}{x-\lambda_i}\right),\]
where $\lambda_i$ are the eigenvalues of $A$ and $u_j$ are the corresponding eigenvectors.
The polynomial $p_{A+\v\v^T}(x)$ has two kinds of zeros $\lambda$:
\begin{enumerate}
\item Those for which $p_A(\lambda)=0$. These are equal to the eigenvalues
$\lambda_j$ of $A$ for
  which the added vector $\v$ is orthogonal to the corresponding eigenvector
  $u_j$, and which do not therefore `move' upon adding $\v\v^T$.
\item Those for which $p_A(\lambda)\neq 0$ and 
\[ f(\lambda)=\left(1-\sum_j\frac{\langle
\v,u_j\rangle^2}{\lambda-\lambda_j}\right)=0.\]
  These are the eigenvalues which have moved and strictly interlace the old
    eigenvalues. 
  The above equation immediately suggests a simple physical model which gives
    intuition as to where these new eigenvalues are located. 
  \begin{figure}[htp]\label{fig:phys}
  \centering
  \includegraphics[scale=.65]{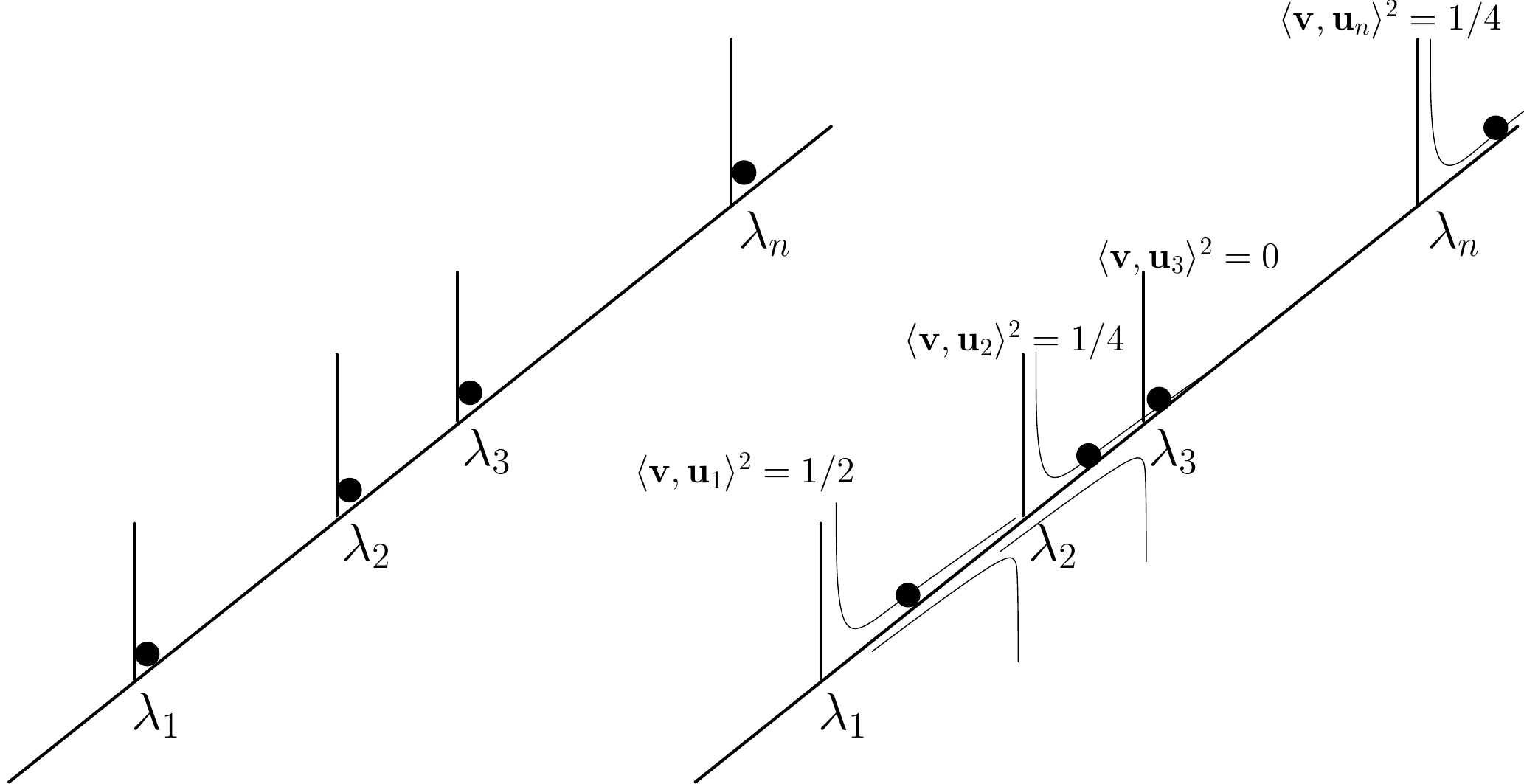}
  \caption{Physical model of interlacing eigenvalues.}
  \end{figure}

  \noindent \textbf{Physical Model.} We interpret the eigenvalues $\lambda$ as charged particles lying on a slope.
  On the slope are $n$ fixed, chargeless barriers located at the initial
    eigenvalues $\lambda_j$, and each particle is resting against one of the
    barriers under the influence of gravity.
  Adding the vector $\v\v^T$ corresponds to placing a charge of $\langle
    \v,u_j\rangle^2$ on the barrier corresponding to $\lambda_j$.
  The charges on the barriers repel those on the eigenvalues with a force
    that is proportional to the charge on the barrier and inversely proportional
    to the distance from the barrier --- i.e., the force from barrier $j$ is given
    by \[\frac{\langle \v,u_j\rangle^2}{\lambda-\lambda_j},\]
    a quantity which is positive for $\lambda_j$ `below' $\lambda$, which are
    pushing the partical `upward', and negative otherwise.
  The eigenvalues move up the slope until they reach an equilibrium in which
    the repulsive forces from the barriers cancel the effect of gravity, which
    we take to be a $+1$ in the downward direction.
  Thus the equilibrium condition corresponds exactly to having the total
  `downward pull' $f(\lambda)$ equal to zero.
\end{enumerate}

With this physical model in mind, we begin to consider what happens to the
  eigenvalues of $A$ when we add a {\em random} vector from our set $\{\v_i\}$.
The first observation is that for any eigenvector $u_j$ (in fact for any vector
  at all), the expected projection of a randomly chosen $\v\in\{\v_i\}_{i\le m}$
  is 
\[ \mathbb{E}_\v \langle \v,u_j\rangle^2 = \frac{1}{m}\sum_i \langle \v_i,u_j\rangle^2 = \frac{1}{m} u_j^T\left(\sum_i
\v_i\v_i^T\right)u_j = \frac{\|u_j\|^2}{m}=\frac{1}{m}.\]
Of course, this does not mean that there is any single vector $\v_i$ in
  our set that realizes this `expected behavior' of equal projections on the
  eigenvectors.
But if we were to add such a vector
\footnote{For concreteness, we remark that this `average' vector would be
  precisely
\[ \v_\textsf{avg} = \frac{1}{\sqrt{m}}\sum_j u_j.\]}
 in our physical model, we would add equal 
  charges of $1/m$ to each of the barriers, and we would expect all of the
  eigenvalues of $A$ to drift forward `steadily'.
In fact, one might expect that after sufficiently many iterations of this
  process, the eigenvalues would all march forward together, with no eigenvalue
  too far ahead or too far behind, and we would 
  end up in a position where $\lambda_{max}/\lambda_{min}$ is bounded.

In fact, this intuition turns out to be correct. 
Adding a vector with equal projections changes the characteristic polynomial 
  in the following manner:
\[ p_{A+\v_\textsf{avg}\v_\textsf{avg}^T}(x)=p_A(x)\left(1-\sum_j\frac{1/m}{x-\lambda_j}\right) =
p_A(x)-(1/m)p_A'(x),\]
since $p_A'(x) = \sum_j \prod_{i\neq j} (x-\lambda_i)$. 
If we start with $A=0$, which has characteristic polynomial $p_0(x)=x^n$,  
  then after $k$ iterations of this process we obtain the polynomial
  \[ p_{k}(x) = (I-(1/m)D)^k x^n\]
where $D$ is the derivative with respect to $x$.
Fortunately, iterating the operator $(I-\alpha D)$ for any
$\alpha>0$
  generates a standard family of orthogonal polynomials -- the {\em associated
  Laguerre polynomials} \cite{dette}. 
These polynomials are very well-studied and the locations of their zeros are
  known; in particular, after $k=dn$ iterations the ratio of the largest to the
  smallest zero is known \cite{dette} to be
  \[ \frac{d+1+2\sqrt{d}}{d+1-2\sqrt{d}},\]
  which is exactly what we want.

To prove the theorem, we will show that we can choose a sequence of actual vectors
  that realizes the expected behavior (i.e. the behavior of repeatedly adding $\v_\textsf{avg}$),
  as long as we are allowed to add 
  arbitrary fractional amounts of the $\v_i\v_i^T$ via the weights $s_i\ge 0$.
We will control the eigenvalues of our matrix by maintaining two barriers as 
  in the physical model, and keeping the eigenvalues between them. 
The lower barrier will `repel' the eigenvalues forward; the upper one will make sure
  they do not go too far. 
The barriers will move forward at a steady pace.
By maintaining that the total `repulsion' at every step of this process is
  bounded, we will be able to guarantee that there is always some multiple of 
  a vector to add that allows us to continue the process.

\subsection{Proof by Barrier Functions}
We begin by defining two `barrier' potential
  functions which measure the quality of the eigenvalues of a matrix. 
These potential functions are inspired by the inverse law 
  of repulsion in the physical model discussed in the last section.
\begin{definition}
For $u,l\in\R$ and $A$ a symmetric matrix with eigenvalues
$\lambda_1,\lambda_2,\ldots,\lambda_{n}$, define:
\[
   \Phi^u(A)\defeq\tr(uI-A)^{-1}
  =
   \sum_i\frac{1}{u-\lambda_i}
  \quad\textrm{(Upper potential)}.
\] 
\[
  \Phi_l(A)\defeq\tr(A-lI)^{-1}
  =
  \sum_i\frac{1}{\lambda_i-l}
   \quad\textrm{(Lower potential)}.\]
\end{definition}

As long as $A\prec uI$ and $A\succ lI$ (i.e., $\lmax(A)<u$ and $\lmin(A)>l$), these potential functions
  measure how far the eigenvalues of $A$ are from the barriers $u$ and $l$. 
In particular, they blow up as any eigenvalue approaches a
  barrier, since then $uI-A$ (or $A-lI$) approaches a singular matrix.
Their strength lies in that they reflect the locations of all the
  eigenvalues simultaneously: for instance, 
  $\Phi^{u}(A)\le 1$ implies that
  no $\lambda_i$ is within distance one of $u$, no $2$
  $\lambda_i$'s are at distance $2$, no $k$ are at distance $k$, and so on.
In terms of the physical model, the upper potential $\Phi^u(A)$ is equal to the total repulsion of the eigenvalues of
  $A$ from the upper barrier $u$, while $\Phi_l(A)$ is the analogous quantity for the lower barrier.

To prove the theorem, we will build the sum $\sum_i s_i\ppi_{i}\ppi_{i}^T$ iteratively, adding one vector at a time.
Specifically, we will construct a sequence of matrices 
\[
   0=A^{(0)},A^{(1)},\ldots, A^{(Q)}
\]
along with positive constants%
\footnote{On first reading the paper, we suggest the reader follow the proof with
  the assignment $\epsU = \epsL = 1$, $u_{0} = n$,
  $l_{0} = -n$, $\deltaU = 2$, $\deltaL = 1/3$.
  This will provide the bound $(6d+1)/ (d-1)$, and eliminates
  the need to use Claim~\ref{clm:lowercauchy}.
}
 $u_0,l_0,\deltaU,\deltaL,\epsU$ and $\epsL$ which satisfy the following conditions:
\begin{enumerate}
\item [(a)] Initially, the barriers are at $u=u_0$ and $l=l_0$ and the potentials are 
\[ \Phi^{u_0}(A^{(0)}) = \epsU \quad \text{and} \quad \Phi_{l_0}(A^{(0)})=\epsL.\]
\item [(b)] Each matrix is obtained by a rank-one update of the previous one ---
specifically by adding a positive multiple of an outer product of some $\v_i$.
\[ A^{(q+1)}=A^{(q)}+t\ppi\ppi^{T} \quad\textrm{for some
  $\ppi\in\{\ppi_{i}\}$ and $t\ge 0$.}\]
\item [(c)] If we increment the barriers $u$ and $l$ by $\deltaU$ and $\deltaL$ respectively
at each step, then the upper and lower potentials do not increase. For every
$q=0,1,\ldots Q$,
\[\Phi^{u+\deltaU}(A^{(q+1)})\le\Phi^{u}(A^{(q)})\le\epsU\quad\textrm{for
$u=u_0+q\deltaU$.}\]
\[\Phi_{l+\deltaL}(A^{(q+1)})\le\Phi_{l}(A^{(q)})\le\epsL\quad\textrm{for
$l=l_0+q\deltaL$.}\]
\item [(d)] No eigenvalue ever jumps across a barrier. For every $q=0,1,\ldots Q$,
\[
\lmax(A^{(q)})< u_0+q\deltaU
\quad \text{and} \quad 
\lmin(A^{(q)})> l_0+q\deltaL.
\]
\end{enumerate}
To complete the proof
 we will choose
 $u_0,l_0,\deltaU,\deltaL,\epsU$ and $\epsL$ so that
after $Q=dn$ steps, the condition number of $A^{(Q)}$ is bounded by
\[
 \frac{\lmax(A^{(Q)})}{\lmin(A^{(Q)})}
 \le \frac{u_0+dn\deltaU}{l_0+dn\deltaL}
 = \frac{d+1+2\sqrt{d}}{d+1-2\sqrt{d}}.
\]
By construction, $A^{(Q)}$ is a weighted sum of at most $dn$ of the vectors, as
desired.

The main technical challenge is to show that conditions (b) and (c) can be
satisfied simultaneously --- i.e., that there is always a choice of $\ppi\ppi^{T} $ to
add to the current matrix which allows us to shift {\em both} barriers up by a 
constant without increasing either potential. We achieve this in the following
three lemmas.

The first lemma concerns shifting the upper barrier. If we shift $u$ forward to
  $u+\deltaU$ without changing the matrix $A$, then the upper potential
  $\Phi^u(A)$ decreases since the eigenvalues $\lambda_i$ do not move and $u$ moves away from
  them.
This gives us room to add some multiple of a vector $t\v\v^T$, which will
  move the $\lambda_i$ towards $l$ and increase the potential, counteracting the
  initial decrease due to shifting.
The following lemma quantifies exactly how much of a given $\v\v^T$ we can add
  without increasing the potential beyond its original value before shifting.

\begin{lemma}[Upper Barrier Shift]\label{lem:upperupd} Suppose $\lmax(A)<u$,
  and $\ppi$ is any vector. If 
\begin{align*} \frac{1}{t} &\ge \frac{
  \ppi^{T} ((u+\deltaU)I-A)^{-2}\ppi}
  {\Phi^u(A)-\Phi^{u+\deltaU}(A)}
  +\ppi^{T} ((u+\deltaU)I-A)^{-1}\ppi
  \defeq U_A(\ppi)\end{align*}
then 
\[ \Phi^{u+\deltaU}(A+t\ppi\ppi^{T} )\le\Phi^u(A)\quad\textrm{and $\lmax(A+t\ppi\ppi^{T} )< u+\deltaU$}.\]
That is, if we add $t$ times $\ppi\ppi^{T} $ to $A$ and shift the upper barrier
by $\deltaU$, then we do not increase the upper potential.
\end{lemma}
\noindent We remark that $U_{A} (\ppi)$ is linear in the outer product $\v \v^{T}$.
\begin{proof}
Let $u'=u+\deltaU$. By the Sherman-Morrison formula, we can write the updated
potential as:
\begin{align*}
&\Phi^{u+\deltaU}(A+t\ppi\ppi^{T} )=\tr(u'I-A-t\ppi\ppi^{T} )^{-1}
\\&=\tr\left((u'I-A)^{-1}+\frac{t(u'I-A)^{-1}\ppi\ppi^{T} (u'I-A)^{-1}}{1-t\ppi^{T} (u'I-A)^{-1}\ppi}\right)
\\&=\tr(u'I-A)^{-1}+\frac{t\tr(\ppi^{T} (u'I-A)^{-1}(u'I-A)^{-1}\ppi)}{1-t\ppi^{T} (u'I-A)^{-1}\ppi}
\\&\qquad\textrm{since $\tr$ is linear and $\tr(XY)=\tr(YX)$}
\\&=\Phi^{u+\deltaU}(A)+\frac{t\ppi^{T} (u'I-A)^{-2}\ppi}{1-t\ppi^{T} (u'I-A)^{-1}\ppi}
\\&=\Phi^{u}(A)-(\Phi^{u}(A)-\Phi^{u+\deltaU}(A))
  +\frac{\ppi^{T} (u'I-A)^{-2}\ppi}{1/t-\ppi^{T} (u'I-A)^{-1}\ppi}
\end{align*}

As $U_{A} (\ppi) > \ppi^{T} (u'I-A)^{-1}\ppi$,
  the last term is finite for $1/t \geq U_{A} (\ppi)$.
By now substituting any 
  $1/t\ge U_A(\ppi)$ we find
  $\Phi^{u+\deltaU}(A+t\ppi\ppi^{T} )\le \Phi^u(A)$.
This also tells us that $\lmax(A+t\ppi \ppi ^T) < u + \deltaU$, 
  as if this were not the case, then there would be some positive
  $t' \leq t$ for which $\lmax(A+t' \ppi \ppi ^T) = u + \deltaU$.
But, at such a $t'$, $\Phi^{u+\deltaU} (A + t' \ppi \ppi ^T)$
  would blow up, and we have just established that it is finite.
\end{proof}

The second lemma is about shifting the lower barrier.
Here, shifting $l$ forward to $l+\deltaL$ while keeping $A$ fixed has the
opposite effect --- it 
  increases the lower potential $\Phi_l(A)$ since the barrier $l$ moves 
  towards the eigenvalues $\lambda_i$.
Adding a multiple of a vector $t\v\v^T$ will move the $\lambda_i$ forward and away
  from the barrier, decreasing the potential.
Here, we quantify exactly how much of a given $\v\v^T$ we need to add 
  to compensate for the initial increase from shifting $l$, and return the potential to 
  its original value before the shift. 
\begin{lemma}[Lower Barrier Shift]\label{lem:lowerupd} Suppose $\lmin(A)>l$,
  $\Phi_{l} (A) \leq 1/\deltaL$,
  and $\ppi$ is any vector. If
\begin{align*} 0 < \frac{1}{t} &\le 
\frac{\ppi^{T} (A-(l+\deltaL)I)^{-2}\ppi}
  {\Phi_{l+\deltaL}(A)-\Phi_l(A)}
  -\ppi^{T} (A-(l+\deltaL)I)^{-1}\ppi
  \defeq L_A(\ppi)\end{align*}
then 
\[ \Phi_{l+\deltaL}(A+t\ppi\ppi^{T} )\le\Phi_l(A)
\quad \text{and} \quad 
\lmin(A+t\ppi\ppi^{T} )> l+\deltaL.\]
That is, if we add $t$ times $\ppi\ppi^{T} $ to $A$ and shift the lower barrier by
$\deltaL$, then we do not increase the lower potential.
\end{lemma}
\begin{proof} 
First, observe that $\lmin(A)>l$ and  $\Phi_{l} (A) \leq 1/\deltaL$
  imply that $\lmin (A) > l + \deltaL$.
So, for every $t > 0$, $\lmin(A+t\ppi\ppi^{T} )> l+\deltaL$.

Now proceed as in the proof for the upper potential. Let $l'=l+\deltaL$.
By Sherman-Morrison, we have:
\begin{align*}
&\Phi_{l+\deltaL}(A+t\ppi\ppi^{T} )=\tr(A+t\ppi\ppi^{T} -l'I)^{-1}
\\&=\tr\left((A-l'I)^{-1}-\frac{t(A-l'I)^{-1}\ppi\ppi^{T} (A-l'I)^{-1}}{1+t\ppi^{T} (A-l')^{-1}\ppi}\right)
\\&=\tr(A-l'I)^{-1}-\frac{t\tr(\ppi^{T} (A-l'I)^{-1}(A-l'I)^{-1}\ppi)}{1+t\ppi^{T} (A-l'I)^{-1}\ppi}
\\&=\Phi_{l+\deltaL}(A)-\frac{t\ppi^{T} (A-l'I)^{-2}\ppi}{1+t\ppi^{T} (A-l'I)^{-1}\ppi}
\\&=\Phi_{l}(A)+(\Phi_{l+\deltaL}(A)-\Phi_{l}(A))-\frac{\ppi^{T} (A-l'I)^{-2}\ppi}{1/t+\ppi^{T} (A-l'I)^{-1}\ppi}
\end{align*}
Rearranging shows that $\Phi_{l+\deltaL}(A+t\ppi\ppi^{T} )\le\Phi_l(A)$ when $1/t\le
L_A(\ppi)$.
\end{proof}

The third lemma identifies the conditions under which we can find a single
  $t\v\v^T$ which allows us to maintain both potentials while shifting barriers, and thereby continue the
  process.
The proof that such a vector exists is by an averaging argument, so 
  this can be seen as the step in which we relate the behavior of actual vectors
  to the behavior of the expected vector $\v_\textsf{avg}$.
Notice that the use of variable weights $t$, from which the eventual
  $s_i$ arise, is crucial to this part of the proof.
\begin{lemma}[Both Barriers]\label{lem:shiftboth} 
If 
$\lmax (A) < u$, $\lmin (A) > l$,
$\Phi^u(A)\le \epsU$,
$\Phi_l(A)\le\epsL$, and $\epsU, \epsL, \deltaU$ and $\deltaL$ satisfy
\begin{equation}\label{bothcond} 0\le\frac{1}{\deltaU}+\epsU \le \frac{1}{\deltaL}-\epsL \end{equation}
then there exists an $i$ and positive $t$ for which
\[
 L_A(\ppi_{i})\geq 1/t \geq U_A(\ppi_{i}),
\quad 
\lmax (A + t \ppi_{i} \ppi_{i}^{T}) < u + \deltaU ,
\]\[\quad 
\text{and}
\quad 
\lmin (A + t \ppi_{i} \ppi_{i}^{T}) > l + \deltaL.
\]

 \end{lemma}
\begin{proof}
We will show that
\[
  \sum_i L_A(\ppi_{i})\ge\sum_i U_A(\ppi_{i}),
\]
 from which the claim will follow by Lemmas~\ref{lem:upperupd} and~\ref{lem:lowerupd}.
We begin by bounding
\begin{align*}
&\sum_i U_A(\ppi_{i}) 
=\frac{
 \sum_i \ppi_{i}^{T} ((u+\deltaU)I-A)^{-2}\ppi_{i}}
  {\Phi^u(A)-\Phi^{u+\deltaU}(A)}
  +\sum_i \ppi_{i}^{T} ((u+\deltaU)I-A)^{-1}\ppi_{i}
\\&=\frac{
   ((u+\deltaU)I-A)^{-2}\bullet(\sum_i\ppi_{i}\ppi_{i}^{T} )}
  {\Phi^u(A)-\Phi^{u+\deltaU}(A)}
  +((u+\deltaU)I-A)^{-1}\bullet\left(\sum_i\ppi_{i}\ppi_{i}^{T} \right)
\\&=\frac{
   \tr((u+\deltaU)I-A)^{-2}}
  {\Phi^u(A)-\Phi^{u+\deltaU}(A)}
  +\tr((u+\deltaU)I-A)^{-1}
\\&\qquad\textrm{since $\sum_i \ppi_{i}\ppi_{i}^{T} =I$ and $X\bullet I=\tr(X)$}
\\&= \frac{
 \sum_i (u+\deltaU-\lambda_i)^{-2}
}{
 \sum_i (u-\lambda_i)^{-1} -\sum_i (u+\deltaU-\lambda_i)^{-1}}
  +\Phi^{u+\deltaU}(A)
\\&= \frac{
 \sum_i (u+\deltaU-\lambda_i)^{-2}}
  {\deltaU \sum_i (u-\lambda_i)^{-1}(u+\deltaU-\lambda_i)^{-1}}
  +\Phi^{u+\deltaU}(A)
\\&\le 1/\deltaU
  +\Phi^{u+\deltaU}(A),
\\&\quad\textrm{as $\sum_i (u-\lambda_i)^{-1}(u+\deltaU-\lambda_i)^{-1}\ge
\sum_i (u+\deltaU-\lambda_i)^{-2}$}
\\&\le 1/\deltaU+\Phi^{u}(A)\le 1/\deltaU+\epsU.
\end{align*}

On the other hand, we have
\begin{align*}
&\sum_i L_A(\ppi_i) 
=\frac{
 \sum_i \ppi_{i}^{T} ((A-(l+\deltaL))^{-2}\ppi_{i}}
  {\Phi_{l+\deltaL}(A)-\Phi_l(A)}
  -\sum_i \ppi_{i}^{T} (A-(l+\deltaL)I)^{-1}\ppi_{i}
\\&=\frac{
   (A-(l+\deltaL)I)^{-2}\bullet(\sum_i\ppi_{i}\ppi_{i}^{T} )}
  {\Phi_{l+\deltaL}(A)-\Phi_l(A)}
  -(A-(l+\deltaL)I)^{-1}\bullet\left(\sum_i\ppi_{i}\ppi_{i}^{T} \right)
\\&=\frac{
   \tr(A-(l+\deltaL)I)^{-2}}
  {\Phi_{l+\deltaL}(A)-\Phi_l(A)}
  -\tr(A-(l+\deltaL)I)^{-1}
\\&\qquad\textrm{since $\sum_i \ppi_{i}\ppi_{i}^{T} =I$ and $X\bullet I=\tr(X)$}
\\&=\frac{
 \sum_i (\lambda_i-l-\deltaL)^{-2}}
  {\sum_i (\lambda_i-l-\deltaL)^{-1}- \sum_i (\lambda_i-l)^{-1}
  }
  -\sum_i (\lambda_i-l-\deltaL)^{-1}.
\\
& \geq 1/ \deltaL -\sum_i (\lambda_i-l)^{-1}
 = 1/\deltaL - \epsilon_L,
\end{align*}
by Claim~\ref{clm:lowercauchy}. 

Putting these together, we find that 
\[ \sum_i U_A(\ppi_{i})\le
\frac{1}{\deltaU}+\epsU\le\frac{1}{\deltaL}-\epsL\le\sum_i
L_A(\ppi_{i}),\] as desired.
\end{proof}

\begin{claim}\label{clm:lowercauchy}
If $\lambda_{i} > l$ for all $i$,
   $0 \leq \sum_i (\lambda_{i} - l)^{-1} \leq \epsL$,
  and $1/\deltaL - \epsL  \geq 0$,
then
\[\frac{\sum_i (\lambda_i-l-\deltaL)^{-2}}
  {\sum_i (\lambda_i-l-\deltaL)^{-1}-\sum_i (\lambda_i-l)^{-1} }
  -\sum_i \frac{1}{\lambda_i-l-\deltaL} \]
\begin{equation}\label{eqn:lowercauchy}
  \ge
 \frac{1}{\deltaL}-\sum_i\frac{1}{\lambda_i-l}.
\end{equation}
\end{claim}
\begin{proof}
We have 
\[
\deltaL \leq 1/\epsL \leq \lambda_{i} - l,
\]
for every $i$.
So, the denominator of the left-most term on the left-hand side is positive,
 and the claimed inequality is equivalent to
\begin{align*}
\sum_i (\lambda_i-l-\deltaL)^{-2}
&\ge 
\left( {\sum_i \frac{1}{\lambda_i-l-\deltaL}-\sum_i\frac{1}{\lambda_i-l}}\right)
\left(  
 \frac{1}{\deltaL}+\sum_i
 \frac{1}{\lambda_i-l-\deltaL}-\sum_i\frac{1}{\lambda_i-l}\right)
\\&=
\left(\deltaL\sum_i\frac{1}{(\lambda_i-l-\deltaL)(\lambda_i-l)}\right)
\left(  
 \frac{1}{\deltaL}+\deltaL\sum_i\frac{1}{(\lambda_i-l-\deltaL)(\lambda_i-l)}
\right)
\\&=
\sum_i\frac{1}{(\lambda_i-l-\deltaL)(\lambda_i-l)}
+ \left(  
\deltaL\sum_i\frac{1}{(\lambda_i-l-\deltaL)(\lambda_i-l)} \right)^2,
  \end{align*}
which, by moving the first term on the RHS to the LHS, is just
\begin{align*} 
&\deltaL\sum_i \frac{1}{(\lambda_i-l-\deltaL)^2(\lambda_i-l)}
\ge
\left(  \deltaL\sum_i\frac{1}{(\lambda_i-l-\deltaL)(\lambda_i-l)}
\right)^2.\end{align*}
By Cauchy-Schwartz,
\begin{align*}
\left(  \deltaL\sum_i\frac{1}{(\lambda_i-l-\deltaL)(\lambda_i-l)} \right)^2
&\le 
  \left( \deltaL\sum_i\frac{1}{\lambda_i-l}\right)\left(\deltaL\sum_i\frac{1}{(\lambda_i-l-\deltaL)^2(\lambda_i-l)}\right)
\\
&\le
\left( \deltaL\epsL \right)\left(\deltaL\sum_i\frac{1}{(\lambda_i-l-\deltaL)^2(\lambda_i-l)}\right)
\\&\qquad 
 \text{since
$\sum (\lambda_{i} - l)^{-1} \leq \epsL$}\\
&\le
1\left(\deltaL\sum_i\frac{1}{(\lambda_i-l-\deltaL)^2(\lambda_i-l)}\right),
\\&\qquad  \text{since }\frac{1}{\deltaL}-\epsL\ge 0,
\end{align*}
and so (\ref{eqn:lowercauchy}) is established.
\end{proof}

\begin{proof}[Proof of Theorem \ref{thm:linalg}] All we need to do now is set $\epsU, \epsL, \deltaU$,
and $\deltaL$ in a manner that satisfies Lemma \ref{lem:shiftboth} and gives a good
bound on the condition number. Then, we can take $A^{(0)}=0$ and 
construct $A^{(q+1)}$ from $A^{(q)}$ by choosing any vector $\ppi_{i}$ with
\[ L_{A^{(q)}}(\ppi_{i})\ge U_{A^{(q)}}(\ppi_{i})\] (such a vector is guaranteed to exist by Lemma
\ref{lem:shiftboth})
and setting $A^{(q+1)}=A^{(q)}+t\ppi_{i}\ppi_{i}^{T} $ for any $t\ge 0$ satisfying:
\[ L_{A^{(q)}}(\ppi_{i})\ge\frac{1}{t}\ge U_{A^{(q)}}(\ppi_{i}).\]

It is sufficient to take
\begin{alignat*}{3}
\deltaL  & = 1 \qquad 
& \epsL  & = \frac{1}{\sqrt{d}} \qquad 
& l_{0} & = -n / \epsL\\
\deltaU & = \frac{\sqrt{d} + 1}{\sqrt{d}-1} \qquad 
& \epsU & = \frac{\sqrt{d} - 1}{d + \sqrt{d}} \qquad 
& u_{0} & = n/\epsU.
\end{alignat*}

We can check that:
\begin{align*}
\frac{1}{\deltaU}+\epsU &=
  \frac{\sqrt{d}-1}{\sqrt{d}+1} + \frac{\sqrt{d}-1}{\sqrt{d} (\sqrt{d}+1)}
 = 1 - \frac{1}{\sqrt{d}}
 =\frac{1}{\deltaL}-\epsL
\end{align*}
so that (\ref{bothcond}) is satisfied.

The initial potentials are
$\Phi^{\frac{n}{\epsU}}(0) = \epsU$ and
$\Phi_{\frac{n}{\epsL}}(0)= \epsL$. After $dn$
steps, we have
\begin{align*}
\frac{\lmax(A^{(dn)})}{\lmin(A^{(dn)})}&\le
\frac{n/\epsU+dn\deltaU}{-n/\epsL+dn\deltaL}
 \\&= \frac{
        \frac{d + \sqrt{d}}{\sqrt{d}-1} + d\frac{\sqrt{d}+1}{\sqrt{d}-1}
      }{
        d - \sqrt{d}
      }
\\&= \frac{d+2\sqrt{d}+1}{d-2\sqrt{d}+1},
\end{align*}
as desired. 
\end{proof}

To turn this proof into an algorithm, one must first compute the vectors $\ppi_{i}$,
  which can be done in time $\bigO{n^{2} m}$.
For each iteration of the algorithm, we must compute 
  $((u + \deltaU)I - A)^{-1}$,  $((u + \deltaU)I - A)^{-2}$,
  and the same matrices for the lower potential function.
This computation can be performed in time $\bigO{n^{3}}$.
Finally, we can decide which edge to add in each iteration by computing
  $U_{A} (\ppi_{i})$ and $L_{A} (\ppi_{i})$ for each edge, which can be done
  in time $\bigO{n^{2} m}$.
As we run for $d n$ iterations, the total time of the algorithm is
  $\bigO{d n^{3} m}$.

\section{Sparsifiers of the Complete Graph}\label{sec:expanders}
Let $G = (V,E)$ be the complete graph on $n$ vertices,
  and let $H = (V,F,w)$ be a weighted graph of
  average degree $d$ that $(1+\epsilon )$-approximates
  $G$.
As $x^{T} L_{G} x = n \norm{x}^{2}$ for every $x$ orthogonal
  to $\bvec{1}$, it is immediate that every vertex of $H$
  has weighed degree between $n$ and $(1+\epsilon) n$.
Thus, one should think of $H$ as being an expander graph in which
  each edge weight has been multiplied by $n/d$.

As $H$ is weighted and can be irregular, it may at first seem
  strange to view it as an expander.
However, it may easily be shown to have the properties that define
   expanders: it has high edge-conductance,
  random walks mix rapidly on $H$ and converge to an almost-uniform distribution,  and it satisfies the
  Expander Mixing Property (see~\cite{AlonChung} 
  or~\cite[Lemma 2.5]{ExpanderSurvey}).
High edge-conductance and rapid mixing would not be so interesting if
  the weighted degrees were not nearly uniform ---
for example, the star graph has both of these properties,
but the random walk on the star graph converges to a very non-uniform
  distribution, and the star does not satisfy the Expander Mixing Property.
For the convenience of the reader, we include a proof that $H$
  has the Expander Mixing Property below.
\begin{lemma}\label{lem:mixing}
Let $L_{H} = (V,E,w)$ be a graph
  that $( 1+\epsilon )$-approximates $L_{G}$,
  the complete graph on $V$.
Then, for every pair of disjoint sets $S$ and $T$,
\[
  \sizeof{w (S,T) - \left(1 + \frac{\epsilon}{2} \right) \sizeof{S} \sizeof{T}} 
\leq n (\epsilon/2)  \sqrt{\sizeof{S} \sizeof{T}},
\]
where $w (S,T)$ denotes the sum of the weights of edges
  between $S$ and $T$.
\end{lemma}
\begin{proof}
We have
\[
  - \frac{\epsilon}{2} L_{G} 
\preceq
  L_{H} - \left(1+\frac{\epsilon}{2} \right) L_{G}
\preceq
  \frac{\epsilon}{2} L_{G},
\]
so
  we can write 
\[
L_{H} =  \left(1+\frac{\epsilon}{2} \right) L_{G} + M,
\]
  where $M$ is a matrix of norm at most 
 $(\epsilon /2) \norm{L_{G}} \leq n \epsilon /2$.
Let $x$ be the characteristic vector of $S$, and let
  $y$ be the characteristic vector of $T$.
We have
\[
  - w (S,T) = x^{T} L_{H} y.
\]
As $G$ is the complete graph and $S$ and $T$ are disjoint, we also know
\[
  x^{T} L_{G} y = - \sizeof{S} \sizeof{T}.
\]
Thus,
\begin{align*}
 x^{T} L_{H} y 
& = 
  \left(1+\frac{\epsilon}{2} \right) x^{T} L_{G} y
  + x^{T} M y\\
& = 
 - \left(1+\frac{\epsilon}{2} \right) \sizeof{S} \sizeof{T}
  + x^{T} M y.
\end{align*}
The lemma now follows by observing that
\[
 x^{T} M y
\leq 
 \norm{M} \norm{x} \norm{y}
\leq 
 n (\epsilon /2) \sqrt{\sizeof{S} \sizeof{T}}.
\]
\end{proof}

Using the proof of the lower bound on the spectral gap
  of Alon and Boppana (see~\cite{Nilli}) one can show that a $d$-regular
  unweighted graph cannot $\kappa$-approximate a complete graph for
  $\kappa$ asymptotically better than \eqref{eqn:ramanujanBound}.
We conjecture that this bound also holds for weighted graphs of average degree $d$.
Presently, we prove the following
  weaker result for such graphs.

\begin{proposition}\label{pro:lowerComplete}
Let $G$ be the complete graph on vertex set $V$,
  and let $H = (V,E,w)$ be a weighted graph with $n$ vertices and
  a vertex of degree $d$.
If $H$ $\kappa$-approximates $G$, then
\[
\kappa  
\geq 
1 + \frac{2}{\sqrt{d}} -
\bigO{\frac{\sqrt{d}}{n}}.
\]
\end{proposition}

\begin{proof}
We use a standard approach.
Suppose $H$ is a $\kappa$-approximation of the complete graph. 
We will construct vectors $x^*$ and $y^*$ orthogonal to the $\bvec{1}$ vector so
  that 
  \[
    \frac{y^{*T}L_Hy^*}{x^{*T}L_Hx^*}
    \frac{\|x^*\|^2}{\|y^*\|^2}
  \]
  is large, and this will give us a lower bound on $\kappa$.

Let $v_{0}$ be the vertex of degree $d$, and let
  its neighbors be $v_{1}, \dotsc , v_{d}$. 
Suppose $v_i$ is connected to $v_0$ by an edge of weight $w_i$, and the total weight
  of the edges between $v_i$ and vertices other than $v_0,v_1,\ldots,v_d$ is
  $\delta_i$.
We begin by considering vectors $x$ and $y$ with
\begin{alignat*}{2}
  x (u) = \begin{cases}
1  & \text{for $u = v_{0}$},
\\
1/\sqrt{d}  & \text{for $u = v_{i}$, $i \geq 1$},
\\
0 & \text{for $u \not \in \setof{v_{0}, \dots , v_{d}}$}
\end{cases}
\end{alignat*}
\begin{alignat*}{2}
  y (u) = \begin{cases}
1  & \text{for $u = v_{0}$},
\\
- 1/\sqrt{d}  & \text{for $u = v_{i}$, $i \geq 1$},
\\
0 & \text{for $u \not \in \setof{v_{0}, \dots , v_{d}}$}
\end{cases}
\end{alignat*}
These vectors are not orthogonal to $\bvec{1}$, but we will take care of that
  later.
It is easy to compute the values taken by the quadratic form at $x$ and $y$:
\begin{align*}
  x^T L_H x &=
  \sum_{i=1}^d w_i(1-1/\sqrt{d})^2 + \sum_{i=1}^d\delta_i(1/\sqrt{d}-0)^2
  \\&= \sum_{i=1}^d w_i +\sum_{i=1}^d (\delta_i + w_{i})/d - 2\sum_{i=1}^d w_i/\sqrt{d}
\end{align*}
and 
\begin{align*}
  y^T L_H y &=
  \sum_{i=1}^d w_i(1+1/\sqrt{d})^2 + \sum_{i=1}^d\delta_i(-1/\sqrt{d}-0)^2
  \\&= \sum_{i=1}^d w_i +\sum_{i=1}^d (\delta_i + w_{i})/d + 2\sum_{i=1}^d w_i/\sqrt{d}.
\end{align*}
The ratio in question is thus
\begin{align*}
  \frac{
    y^TL_Hy
   }{
     x^TL_Hx
   } 
   & = \frac{
     \sum_i w_i + \sum_i (\delta_i + w_{i})/d + 2\sum_i w_i/\sqrt{d}
    }{
     \sum_i w_i + \sum_i (\delta_i + w_{i})/d - 2\sum_i w_i/\sqrt{d}
   }
   \\
   & = \frac{
     1+\frac{1}{\sqrt{d}}\frac{2\sum_i w_i}{\sum_i w_i + \sum_i (\delta_i + w_{i})/d}
    }{
     1-\frac{1}{\sqrt{d}}\frac{2\sum_i w_i}{\sum_i w_i + \sum_i (\delta_i + w_{i})/d}
    }.
\end{align*}
Since $H$ is a $\kappa$-approximation, all weighted degrees must lie between $n$
and $n\kappa$, which gives
\[
  \frac{
    2\sum_{i} w_{i}
  }{
    \sum_{i} w_{i} +\sum_i (\delta_i + w_{i})/d
  }
  = \frac{2}
    {1+\frac{
          \sum_i (\delta_i + w_{i})/d
	 }{
	  \sum_i w_i
	 }
    }
  \ge \frac{2}{1+\kappa}.
\]
Therefore,
\begin{equation}\label{eqn:testvectors} 
\frac{
  y^TL_Hy
}{
  x^TL_Hx
} \ge 
\frac{
  1+\frac{1}{\sqrt{d}}\frac{2}{1+\kappa}
  }{
  1-\frac{1}{\sqrt{d}}\frac{2}{1+\kappa}
}.
\end{equation}
Let $x^*$ and $y^*$ be the projections of $x$ and $y$ respectively orthogonal to
the $\bvec{1}$ vector. Then 
\[
  \|x^*\|^2 = \|x\|^2-\langle x,\bvec{1}/\sqrt{n}\rangle^2 =
  2-\frac{(1+\sqrt{d})^2}{n}
\]
and 
\[
  \|y^*\|^2 = \|y\|^2-\langle y,\bvec{1}/\sqrt{n}\rangle^2 =
  2-\frac{(1-\sqrt{d})^2}{n}
\]
so that as $n\to\infty$ 
\begin{equation}\label{eqn:projnorms}
\frac{\|x^*\|^2}{\|y^*\|^2} = 1 - \bigO{\frac{\sqrt{d}}{n}}.
\end{equation}
Combining (\ref{eqn:testvectors}) and (\ref{eqn:projnorms}), we conclude that asymptotically:
\[
  \frac{y^{*T}L_Hy^*}{x^{*T}L_Hx^*}
  \frac{\|x^*\|^2}{\|y^*\|^2}
  \ge 
\frac{
  1+\frac{1}{\sqrt{d}}\frac{2}{1+\kappa}
  }{
  1-\frac{1}{\sqrt{d}}\frac{2}{1+\kappa}.
}
\left( 1 - \bigO{\frac{\sqrt{d}}{n}} \right)
\]
But by our assumption the LHS is at most $\kappa$, so we have
\[
  \kappa \ge
\frac{
  1+\frac{1}{\sqrt{d}}\frac{2}{1+\kappa}
  }{
  1-\frac{1}{\sqrt{d}}\frac{2}{1+\kappa}.
}
\left( 1 - \bigO{\frac{\sqrt{d}}{n}} \right)
\]
which on rearranging gives
\[ 
  \kappa \ge 1+\frac{2}{\sqrt{d}} - \bigO{\frac{\sqrt{d}}{n}}
\]
as desired.
\end{proof}
\section{Conclusion}
We conclude by drawing a connection between Theorem \ref{thm:linalg} and an
  outstanding open problem in mathematics, the Kadison-Singer conjecture. 
This conjecture, which dates back to 1959, is equivalent to the well-known
  Paving Conjecture \cite{ak-and, cass} as well as to a stronger form of the restricted invertibility
  theorem of Bourgain and Tzafriri \cite{btz,cass}.
The following formulation is due to Nik Weaver \cite{weaver}.

\begin{conjecture} \label{conj:ks} There are universal constants 
$\epsilon > 0, \delta > 0$, and $r\in \mathbb{N}$ for which the following
statement holds. If $\ppi_1,\ldots, \ppi_m\in\R^n$
satisfy $\|\ppi_i\|\le \delta$ for all $i$ and
\[\sum_{i\le m}\ppi_i\ppi_i^T = I,\]
  then there is a partition $X_1,\ldots X_r$ of $\{1,\ldots, m\}$ for which
\[\left\|\sum_{i\in X_j}\ppi_i\ppi_i^T\right\|\le 1-\epsilon\]
  for every $j=1,\ldots, r$.
\end{conjecture}

Suppose we had a version of Theorem \ref{thm:linalg} which, assuming 
$\|\ppi_i\|\le\delta$, guaranteed that the
  scalars $s_i$ were all either $0$ or some constant $\beta > 0$, and 
  gave a constant approximation factor $\kappa < \beta$.
Then we would have
\[ I \preceq \beta \sum_{i\in S} \ppi_i\ppi_i^T\preceq \kappa\cdot I,\]
 for $S=\{i:s_i\ne 0\}$, yielding
a proof of Conjecture \ref{conj:ks} with 
    $r=2$ and $\epsilon = \min \{1-\frac{\kappa}{\beta}, \frac{1}{\beta}\}$ since
    \[ \left\|\sum_{i\in {S}} \ppi_i\ppi_i^T\right\| \le \frac{\kappa}{\beta} \le
    1-\epsilon\quad\textrm{and}\]
    \[ \left\|\sum_{i\in\overline{S}} \ppi_i\ppi_i^T\right\| =
    1-\lambda_\textrm{min}\left(\sum_{i\in S}\ppi_i\ppi_i^T\right) \le
    1-\frac{1}{\beta} \le 1-\epsilon.\]

As a special case, such a theorem would also imply the existence of {\em unweighted} sparsifiers
  for the complete graph and other (sufficiently dense) edge-transitive graphs. 
It is also worth noting that the $\|\ppi_i\|\le\delta$ condition when applied to vectors $\{\Pi_e\}_{e\in E}$ arising from a graph 
  simply means that the effective resistances
  of all edges are bounded; thus, we would be able to conclude that any graph with sufficiently small resistances
  can be split into two graphs that approximate it spectrally.

\end{document}